\documentclass[10pt,conference]{IEEEtran}
\usepackage{times}

\usepackage{amsmath}
\usepackage{amsfonts}
\usepackage{latexsym}
\usepackage{amssymb}

\usepackage{upref}
\usepackage{theorem}
\usepackage{graphicx}
\usepackage{psfrag}
\usepackage{multirow}
\usepackage{epstopdf}

\usepackage{txfonts}
\usepackage{color}
\usepackage{enumerate}

\usepackage{comment}
\usepackage{hhline}
\usepackage{tabu}

\usepackage{soul}
\usepackage[ruled,noline]{algorithm2e}

\theoremstyle{plain}
\theorembodyfont{\normalfont\slshape}

\newtheorem{thm}{Theorem$\!$}
\newenvironment{theorem}
{\begin{thm}\hspace*{-1ex}{\bf.}}{\end{thm}}

\newtheorem{lem}[thm]{Lemma$\!$}
\newenvironment{lemma}{\begin{lem}\hspace*{-1ex}{\bf.}}{\end{lem}}

\newtheorem{prop}[thm]{Proposition$\!$}
\newenvironment{proposition}{\begin{prop}\hspace*{-1ex}{\bf.}}{\end{prop}}

\newtheorem{cor}[thm]{Corollary$\!$}
\newenvironment{corollary}{\begin{cor}\hspace*{-1ex}{\bf.}}{\end{cor}}

\newtheorem{defn}{Definition$\!$}
\newenvironment{definition}{\begin{defn}\hspace*{-1ex}{\bf.}}{\end{defn}}

\newtheorem{xmpl}{Example$\!$}
\newenvironment{example}{\begin{xmpl}\hspace*{-1ex}{\bf.}}{\end{xmpl}}

\newtheorem{const}{Construction$\!$}
\newenvironment{construction}{\begin{const}\hspace*{-1ex}{\bf.}}{\end{const}}

\newtheorem{prob}[thm]{Problem$\!$}

\newcounter{enumrom}
\renewcommand{\theenumrom}{(\roman{enumrom})}

\makeatletter
\renewcommand{\@endtheorem}{\endtrivlist}
\makeatother

\makeatletter
\renewcommand{\thefigure}{{\@arabic\c@figure}}
\renewcommand{\fnum@figure}{{\bf Figure\,\thefigure}}
\makeatother

\renewcommand{\QEDclosed}{\mbox{\rule[-1pt]{1.3ex}{1.3ex}}} %

\renewcommand{\baselinestretch}{0.92}

\newcommand{\cC}{{\cal C}}

\newcommand{\cF}{{\cal F}}

\newcommand{\cM}{{\cal M}}
\newcommand{\cN}{{\cal N}}

\newcommand{\cR}{{\cal R}}

\newcommand{\be}[1]{\begin{equation}\label{#1}}
\newcommand{\ee}{\end{equation}}

\renewcommand{\leq}{\leqslant}

\renewcommand{\geq}{\geqslant}

\newcommand{\Cref}[1]{Co\-ro\-lla\-ry\,\ref{#1}}

\newcommand{\farea}{Z_1}
\newcommand{\sarea}{Z_2}
\newcommand{\iarea}{Z_i}

\DeclareMathAlphabet{\mathbfsl}{OT1}{cmr}{bx}{it}

\newcommand{\vect}[1]{\boldsymbol{#1}}

\newcommand{\cellnum}{n}

\newcommand{\seqN}{L}

\mathchardef\mhyphen="2D
\newcommand{\wi}{\ensuremath{d}-imbalance }
\newcommand{\wiAB}{\ensuremath{d}\mhyphen imb }
\newcommand{\betag}{\tilde{\beta}}

\outer\def\proclaim #1. #2\par{\medbreak
 \noindent{\bf#1.\enspace}{\sl#2\par}%
 \ifdim\lastskip<\medskipamount \removelastskip\penalty55\medskip\fi}

\mathchardef\inn="3232
\renewcommand{\in}{{\,\inn\,}}

\newcommand\yc[1]{{\color{magenta}#1}}

\begin{document}

\sloppy

\title{$d$-imbalance WOM Codes for Reduced Inter-Cell Interference in Multi-Level NVMs}

\author{ \authorblockN{ \textbf{Evyatar Hemo} and \textbf{Yuval Cassuto}}
\authorblockA{ Department of Electrical Engineering, Technion -- Israel Institute of Technology \\}
{\it  evyatar.hemo@gmail.com, ycassuto@ee.technion.ac.il}
}
\maketitle

\maketitle

\begin{abstract}
In recent years, due to the spread of multi-level non-volatile memories (NVM), $q$-ary write-once memories (WOM) codes have been extensively studied. By using WOM codes, it is possible to rewrite NVMs $t$ times before erasing the cells. The use of WOM codes enables to improve the performance of the storage device, however, it may also increase errors caused by inter-cell interference (ICI). This work presents WOM codes that restrict the imbalance between code symbols throughout the write sequence, hence decreasing ICI. We first specify the imbalance model as a bound $d$ on the difference between codeword levels. Then a $2$-cell code construction for general $q$ and input size is proposed. An upper bound on the write count is also derived, showing the optimality of the proposed construction. In addition to direct WOM constructions, we derive closed-form optimal write regions for codes constructed with continuous lattices.

On the coding side, the proposed codes are shown to be competitive with known codes not adhering to the bounded imbalance constraint. On the memory side, we show how the codes can be deployed within flash wordlines, and quantify their BER advantage using accepted ICI models.
\end{abstract}

\section{Introduction}

In many multi-level non-volatile memory (NVM) technologies there is an inherent asymmetry between increasing and decreasing the level to which a cell is programmed. In particular, in flash memories cell levels are represented by quantities  of electrical charge, and removing charge is known to be much more difficult than adding charge. This asymmetry implies significant access limitations, whereby erasing cells must be done simultaneously in large groups of order $~10^6$ cells (called blocks). From this limitation stem many of the serious performance issues of flash, most prominently low write rates and accelerated cell wear. \\
A possible solution for reducing erasure operations and increasing the lifetime of flash memories is using write-once memory (WOM) codes. The use of WOM codes in flash memories enables multiple writes before executing the costly erasure operation. Indeed, it was shown~\cite{saher} that by using WOM codes in multi-level NVMs, it is possible to reduce write amplification, and thus increase the lifetime of the device. This justifies the recent extensive study of $q$-ary WOM codes~\cite{Bhatia,Yuval,s2,s5,s4} that generalize the original binary WOM model~\cite{Rivest} to multi-level flash. \\
In light of this promise, the main issue holding back WOM codes from deployment seems to be concerns related to inter-cell interference (ICI)~\cite{cai}. Since WOM codes allow updating pages {\em in place} and non-sequentially, there is a potential risk that these updates will disturb adjacent pages. The risk of ICI disturbance becomes more significant as cell levels are updated to much higher levels than their neighbors. \\
Hence in this paper we propose ICI mitigating WOM codes that maintain a degree of balance between the physical levels of the cells throughout the write sequence. As a consequence, the level difference between adjacent memory cells is constrained to be up to an imbalance parameter $d$ chosen for the code.

In Section~\ref{sec:optimal_cnst} we present our main contribution: a $d$-imbalance $2$-cell WOM-code construction that yields codes for general values of $q$ and input sizes. We also derive an upper bound on the number of guaranteed writes given the imbalance parameter, and show that our construction is optimal. With a comparison table we show that the numbers of writes our codes offer are favorable even relative to unconstrained existing codes. The uniqueness of this work over prior ICI codes is that it mitigates ICI within the WOM framework. Whereas known ICI-WOM codes~\cite{prev2} only constrain the transition of individual cells at an individual write, our codes jointly maintain balance between the symbols of the WOM codeword throughout the write sequence.

In section~\ref{sec:EWLbhatia} we pursue \wi WOM codes using the {\em lattice-based} WOM construction technique developed in~\cite{Bhatia,bhatia2,s4}. We derive in closed form the optimal continuous write regions with the \wi constraint for $n=2$ cells and $t=2$ writes. The optimal boundary between the write regions turns out to be a {\em parabola}, while the known optimal boundary in the unconstrained case is known to be a hyperbola~\cite{s4}. Another curious fact we find is that for the \wi case the optimal sum-rate is achieved by constant-rate codes, in contrast to classical unconstrained codes exhibiting a gap between optimal variable- and fixed-rate codes. The information loss from coding in short $2$-cell blocks is minor considering the advantages. For example, it was shown~\cite{Yuval} that when the input size $M$ of the code is order $\sqrt{q}$, the ratio between the $2$-cell code's sum-rate to the (variable-rate) WOM capacity tends to $1$ as $q$ grows.

In Section~\ref{sec:practical} we gear more toward practical realization of the codes and show how multiple WOM codewords can be concatenated in a flash wordline to maintain the \wi constraint globally. In addition, we analyze the improvement in worst-case ICI expected when cells are constrained with the \wi property.

\section{Background and Definitions} \label{sec:def}
\subsection{Inter-cell interference (ICI)}
In flash memories, changing the electrical charge of one floating-gate transistor can change the charge of its neighboring transistors through the parasitic capacitance-coupling effect~\cite{ici-effect}. This effect is referred to as inter-cell interference (ICI), and it is one of the most dominant sources for errors in flash memories. In addition, with the continuing process of scaling down cell sizes, the distance between adjacent cells becomes smaller. As a result, the parasitic capacitance between a cell and its neighbor cells increases, which in turn increases the ICI. \\
Moreover, as was shown in~\cite{berman}, the write process is also a key feature in the ICI mechanism. NAND flash devices commonly use the {\em incremental step pulse program} (ISPP) write method to mitigate cell variability~\cite{ispp}. In the ISPP method each program level induces a sequence of program pulses followed by a verification process to assure proximity to the target level. Each program step increases the voltage level of a cell by $\Delta V_{pp}$, which is significantly smaller than the actual voltage levels representing memory values. As was described in~\cite{berman}, when cells are programmed by ISPP, it is possible to compensate ICI errors in the cells that have not yet reached their target values. If during the write sequence cell $\#1$ causes ICI in cell $\#2$, it can be detected by the verification process of cell $\#2$, and the program steps of cell $\#2$ may be modified to compensate for this ICI. However, when a certain cell already reached its target level, ICI from a neighbor cell cannot be compensated and may cause a write error.\\
The accepted conclusion from the ICI behavior described above is that ICI errors are more likely when the difference between target levels of adjacent cells is high~\cite{berman,ici-ber}. Therefore, a coding scheme that balances voltage levels of adjacent cells, forbidding significant voltage differences, is likely to reduce ICI errors. Detailed analysis of the ICI and its effect on the bit-error rate (BER) appears in Section~\ref{subsec:ici_an}.

\subsubsection*{ICI due to lateral charge spreading}
Recently, a new $3D$ vertical \emph{charge-trap} flash memory was commercially introduced~\cite{samsung}. This flash device was reported to have low ICI from the capacitance-coupling effect, however, it suffers from ICI due to charge migration between adjacent cells, termed as \emph{lateral charge spreading} effect. It was shown in~\cite{mod} that if the level difference between adjacent cells is small, the lateral charge spreading effect is significantly reduced. Therefore, a coding scheme that balances the charge levels of adjacent cells is similarly warranted for this new form of ICI.

\subsection{WOM codes}
Our focus in this paper is on limited-imbalance codes in the WOM model, because the in-place re-writing of WOM codes makes them especially prone to ICI. We first review some necessary background on $q$-ary WOM codes.
\begin{definition}
A fixed rate \textbf{WOM code} $\cC\left(\cellnum,q,t,M\right)$ is a code applied to a size $\cellnum$ block of $q$-ary cells, and guaranteeing $t$ writes of input size $M$ each.
\end{definition}
A WOM code is specified through a pair of functions: the \emph{decoding} and \emph{update} functions.
\begin{definition}
The decoding function is defined as $\psi:\left\{0, \ldots, q-1\right\}^n \rightarrow \left\{0, \ldots, M-1\right\}$, mapping the current levels of the $n$ cells to the data input in the most recent write. The update function is defined as $\mu :\left\{0, \ldots, q-1\right\}^n \times \left\{0, \ldots, M-1\right\} \rightarrow \left\{0, \ldots, q-1\right\}^n$, specifying the new cell levels as a function of the current cell levels and the new data value at the input. By the WOM requirement, the $i$-th cell level output by $\mu$ cannot be lower than the $i$-th cell level in the input.
\end{definition}
\begin{definition}
The code's \textbf{physical state} is defined as the $n$ $q$-ary levels to which the cells are currently programmed. The code's \textbf{logical state} is the data element from $\{0,\ldots,M-1\}$ returned by $\psi$ on the current physical state.
\end{definition}
A {\em write region} spanned from a physical state is a set of physical states accessible from it under the WOM requirement. The size of this set we call the {\em area} of the write region. If at a given physical state the code admits more write(s), then this physical state must span a write region with area at least $M$.
\begin{example} \label{ex1}
Let us consider two sample WOM codes. In Fig.~\ref{fig:Example1} (a) we have the decoding function of $\cC\left(n=2,q=7,t=3,M=8\right)$ constructed by Construction $3$ in~\cite{Yuval}. This code is applied on a pair of $q=7$-level memory cells, enabling $t=3$ guaranteed writes of size $M=8$ each. In Fig.~\ref{fig:Example1} (b) we have a code $\cC\left(n=2,q=7,t=3,M=8\right)$, offering the same number of writes.
\begin{figure}[htbp]
   \centering
   \includegraphics[height=2.1in,keepaspectratio=true,width=0.48\textwidth]{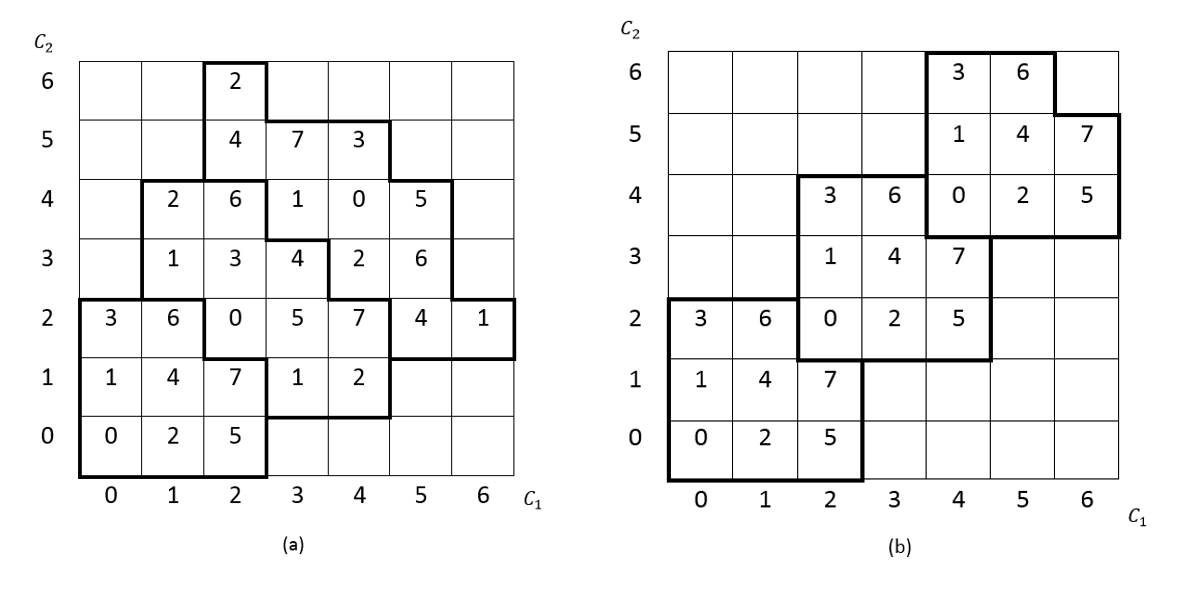}
   \caption{Sample $n=2$ WOM constructions (from \cite{Yuval}). (a) - Decoding function $\psi$ for a code $\cC\left(2,7,3,8\right)$. (b) - Decoding function $\psi$ for another code $\cC\left(2,7,3,8\right)$. Physical states are represented by $\left(c_1,c_2\right)$ and logical states are labeled inside each square.}
   \label{fig:Example1}
\end{figure}
Considering Fig.~\ref{fig:Example1} (a), let us assume we want to perform three writes of the logical states $7$, $6$ and $2$ using this WOM code. For the first write the logical state is $7$ and the physical state is $\left(2,1\right)$. When updating the logical state to $6$, the physical state becomes $\left(2,4\right)$. For the third write of $2$, the physical state becomes $\left(2,6\right)$. After the third write, we reach a physical state with level difference of $4$ between the cells. As a consequence, given that the pair of cells are adjacent, cell $1$ is likely to suffer from ICI. The code in Fig.~\ref{fig:Example1} (b) maintains a better balance between the two cell levels, but will offer fewer writes if extended beyond $q=7$.
\end{example}
In order to reduce ICI, we now define the \wi model for WOM codes.
\begin{definition} \label{def:imbalance}
A \textbf{\wi} WOM code $\cC_{\wiAB}\left(n,q,t,M\right)$ is a WOM code that guarantees that after \textbf{each write} the physical states of the cells $c_i$, $1\leq i \leq n$, must satisfy
\begin{equation} \label{eq:wi-d}
\max_{i,j: i\neq j} | c_i-c_j | \leq d,
\end{equation}
for any write sequence.
\end{definition}
A \wi code guarantees that the level imbalance between cells cannot exceed $d$. Therefore, all cells sustain similar (same up to $d$) levels of charge injection, thus imposing control on the ICI disturbance. When $d=q-1$, we get a standard unconstrained WOM code without balancing properties. As $d$ decreases, the balancing improves, but the added constraints may lead to lower re-write capabilities.
\section{Optimal $d$-Imbalance Construction}\label{sec:optimal_cnst}
Before showing our main construction, we prelude this section with a discussion on which $d$ parameters would be interesting to consider. Given $n$ and $M$, the $d$ imbalance parameter of a code $\cC_{\wiAB}\left(n,q,t,M\right)$ cannot be less than $\left\lceil \sqrt[\cellnum]{M} \right\rceil-1$. That is because in any physical state, at most $(d+1)^n$ states are accessible for the next write while keeping the $d$ imbalance constraint. So to be able to write any of the $M$ values in the next write we must have $M\leq (d+1)^n$.
\begin{example}\label{ex:max_balance}
For $n=2$ and $M=8$, the lowest possible imbalance parameter is $d=2$. It turns out that for this extreme case the simple ``diagonal stacking'' construction of Fig.~\ref{fig:Example1}$\left(b\right)$ is an optimal $\cC_{\ensuremath{2}\mhyphen imb}\left(2,q,t,8\right)$ code with $t=\left\lfloor \left(q-1\right)/2\right\rfloor$ writes. It is straightforward to generalize this construction to produce  $d=(a-1)$-imbalance codes $\cC_{\ensuremath{a-1}\mhyphen imb}\left(2,q,t,a^2-1\right)$, for any $2<a \in \mathbb{Z}$, and providing $t=\left\lfloor \left(q-1\right)/\left(a-1\right)\right\rfloor$ writes.
\end{example}
Requiring maximal balance (minimal $d$) results in weak codes with small numbers of writes. A better tradeoff between balancing and re-write efficiency is obtained when $d$ is relaxed from the extreme value, in which case good balancing (low ICI errors) is achieved while getting more writes. This will be the case we handle in our following construction.
\subsection{Construction}
We now turn to present a construction where $d$ is $1$ larger than in the extreme case of Example~\ref{ex:max_balance}.
\begin{construction}\label{const2}
For any $q$, we define an $n=2$ WOM code with $M=a^2-1$, $2<a \in \mathbb{Z}$, and $d = a$-imbalance parameter as follows.
\begin{enumerate}
  \item \emph{Decoding function:} \\
  The decoding function is specified in Fig.~\ref{fig:WIcode}. The number shown at position $\left(c_1,c_2\right)$ represents the logical state as returned by the decoding function $\psi\left(c_1,c_2\right)$.
  \item \emph{Update function:} \\
  The update function is specified with the aid of $3$ distinctly colored regions in Fig.~\ref{fig:WIcode}, which represent the worst case regions of the $3$ writes. The update function determines the new physical state $(c'_1,c'_2)$ given the current state $\left(c_1,c_2\right)$ and the new value to be written $m$ as follows:
\begin{enumerate}
  \item locate all physical states with $(c''_1,c''_2)\geq \left(c_1,c_2\right)$ element-wise, for which $\psi(c''_1,c''_2)=m$.
  \item $(c'_1,c'_2)$ is chosen as the pair $(c''_1,c''_2)$ that minimizes the sum of coordinates $\left| (c''_1,c''_2) - \left(c_1,c_2\right)\right|$.
\end{enumerate}
The bottom-left region in Fig.~\ref{fig:WIcode} has all $M=a^2-1$ logical states $m_{i,j},\,0\leq i,j \leq a-1$ excluding $i=j=a-1$, accessible for the first write. Each of the other two regions has all the $M=a^2-1$ logical states accessible from every physical state in the region to the left and down. Hence the update function supports any sequence of $3$ written values without exceeding the top-right region.
\end{enumerate}
\begin{figure}[htbp]
   \centering
   \includegraphics[height=2.9in,keepaspectratio=true,width=0.5\textwidth]{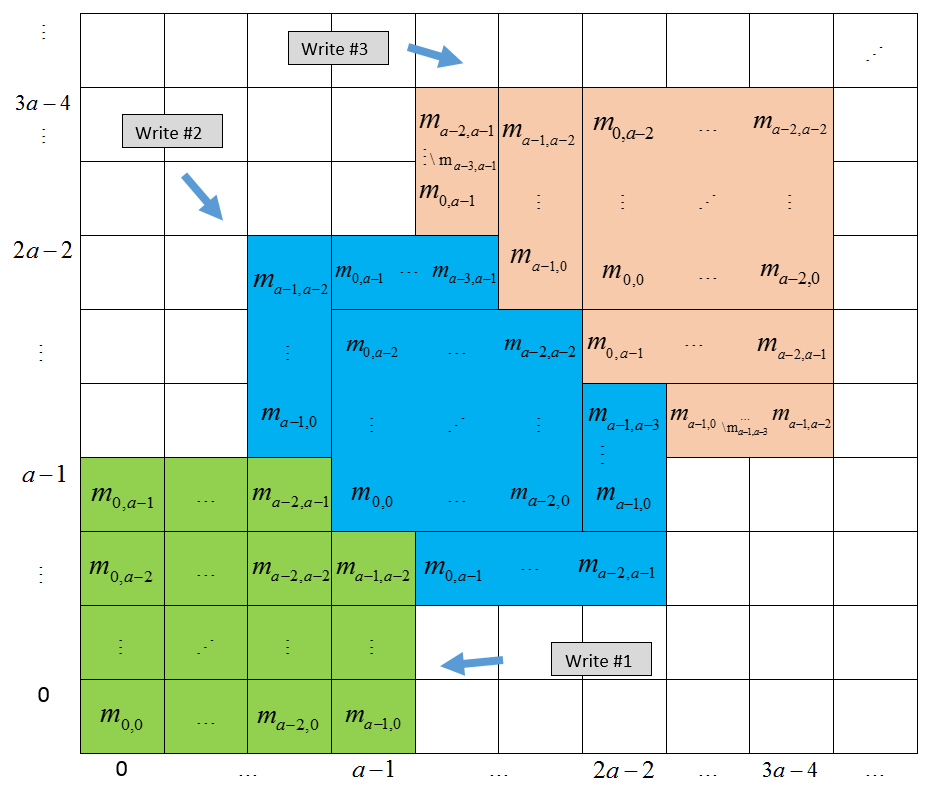}
   \caption{Decoding function of the $d=a$-imbalance WOM code $\cC_{a\mhyphen imb}\left(2,q,t,a^2-1\right)$. The notation $ m_{a-1,0} \underset{\backslash m_{a-1,a-3}}{\cdots} m_{a-1,a-2}$ represents all the logical states $ m_{a-1,0}$ to $ m_{a-1,a-2}$ excluding $m_{a-1,a-3}$. The three colored regions represent the worst case regions of the first three writes of this code. }
   \label{fig:WIcode}
\end{figure}
\end{construction}
We next examine the imbalance parameter of Construction~\ref{const2}. It is straightforward to see from Fig.~\ref{fig:WIcode} that all physical states $(x,y)$ used by the update function satisfy $|y-x|\leq a$, as required.\\
Before discussing the extension of Construction~\ref{const2} beyond $3$ writes, we give an example for the special case $a=3$.

\begin{example}\label{ex:cnst}
In this example, we demonstrate Construction~\ref{const2} for $M=8$ (corresponding to $a=3$), and $q=6$. The $3$ writes of the resulting code $\cC_{3\mhyphen imb}\left(2,6,3,8\right)$ are presented in Fig.~\ref{fig:WIcodeExp}, where the labels $m_{i,j}$ are given by $m_{i,j}=i+ja$.
%
\begin{figure}[htbp]
   \centering
   \includegraphics[height=2.1in,keepaspectratio=true,width=0.5\textwidth]{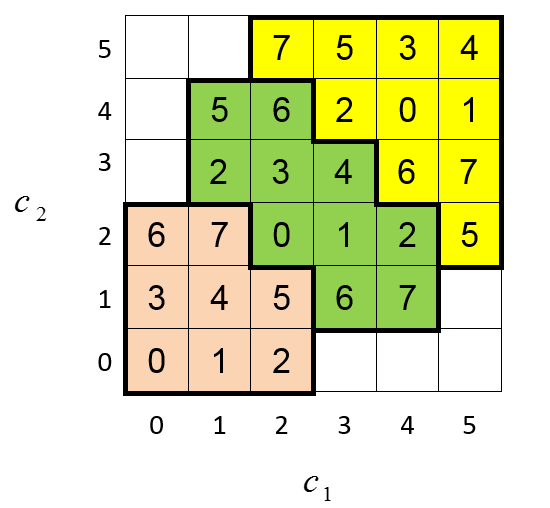}
   \caption{Decoding function and update regions for the $d=3$-imbalance WOM code $\cC_{3\mhyphen imb}\left(2,6,3,8\right)$. The three colored regions represent the worst case regions of the three writes of the code.}
   \label{fig:WIcodeExp}
\end{figure}
It can be checked that for any sequence of $3$ written logical states the update function of Construction~\ref{const2} can succeed without exceeding level $5$ at any of the cells.
\end{example}
\textbf{Extending Construction~\ref{const2} to general $\boldsymbol{q}$.}
To extend the decoding and update functions of Construction~\ref{const2} to general $q$, we copy the three regions of Fig.~\ref{fig:WIcode} and lay out the copies such that the origin of a new copy is placed on the top-right corner of the previous copy. Note that such extension requires us to relabel the logical states along the main diagonal because the origin logical state and the top-right logical state are different ($m_{0,0}$ vs. $m_{a-2,a-2}$). Observe in Fig.~\ref{fig:WIcode} that the logical-state values on the main diagonal are of the form $m_{i,i}$, and that these values do not appear elsewhere in the two-dimensional array. This means that we can lay out these values in a cyclic fashion across copies. That is, at physical state $(j,j)$ in the extended array we place logical value $m_{j\bmod a-1,j\bmod a-1}$. Apart from the main diagonal, the extended copies of the three regions have the same assignment of logical values as the base copy of Fig.~\ref{fig:WIcode}. In Example~\ref{ex:cnst} we place the origin of a second copy at physical state $(5,5)$, and get $3$ more writes while relabeling the main diagonal from $0,4,0,4,0,4$ in the first copy, to $4,0,4,0,4,0$ in the second copy.

We now derive the number of guaranteed writes of a \wi code produced by Construction~\ref{const2}.
\begin{theorem} \label{th:wi-t}
For any $q$ and $2<a \in \mathbb{Z}$, a $d=a$-imbalance WOM code $\cC_{a\mhyphen imb}\left(2,q,t,a^2-1\right)$ constructed by Construction~\ref{const2} guarantees
\begin{equation} \label{eq:wi-t}
t=\left\lfloor\frac{3\left(q-1\right)}{3a-4}\right\rfloor
\end{equation}
writes.
\end{theorem}
\begin{proof}
With the periodic extension of Fig.~\ref{fig:WIcode}, we know that for $t=3\ell$ writes, $\ell$ integer, $q-1=\ell(3a-4)$ is sufficient. Substituting $\ell=t/3$ we get $t=3(q-1)/(3a-4)$ as required. To complete the proof, we need to show~\eqref{eq:wi-t} for $t=3\ell+r$ writes, also for the cases $r=1,2$. In these cases, the last $r$ writes each increases $q$ by $a-1$. Therefore, we have
\begin{equation} \label{eq:q_rcase}
q-1=\ell(3a-4)+r(a-1).
\end{equation}
Substituting $\ell=(t-r)/3$ and rearranging, we get
\begin{equation} \label{eq:t_rcase}
t=\frac{3(q-1)-r}{3a-4}.
\end{equation}
It appears that the expression in the right-hand side of~\eqref{eq:t_rcase} may be smaller than the right-hand side of~\eqref{eq:wi-t}. We show that this cannot happen. From~\eqref{eq:q_rcase} we know that $q-1 \equiv r(a-1) \pmod {3a-4}$. Therefore, $3(q-1) \equiv 3r(a-1)\equiv r(3a-3)\equiv r \pmod {3a-4}$. Now expanding~\eqref{eq:wi-t}, we get
\begin{equation} \label{eq:t_final}
\hspace*{-0.1cm}
\left\lfloor\frac{3\left(q-1\right)}{3a-4}\right\rfloor = \frac{3\left(q-1\right)}{3a-4} - \frac{3(q-1)\bmod (3a-4)}{3a-4}=\frac{3(q-1)-r}{3a-4},
\end{equation}
which proves ~\eqref{eq:wi-t} for all $t$. The fact that the periodic extension of Construction~\ref{const2} has $d=a$ is immediate from Fig.~\ref{fig:WIcode}.
\end{proof}
Substituting into~\eqref{eq:wi-t} the special case $a=3$, $q=6$, given in Example~\ref{ex:cnst}, we indeed get $t=3$.

\subsection{Upper bound on the guaranteed number of writes}
We now derive an upper bound on the number of guaranteed writes of a \wi code that shows that Construction~\ref{const2} gives optimal codes. Optimality will be proved for the special case $a=3$, that is, for codes with $M=8$ and $d=3$ imbalance. A similar technique can extend the upper bound to more general $a$ values. We start with the following definitions and lemmas.
\begin{definition}
If a WOM code $i$-th write starts at state $\left(x_i,y_i\right)$ and ends at state $\left(x_{i+1},y_{i+1}\right)$, then the (non-negative) \textbf{write sum} of the $i$-th write is defined as $x_{i+1}-x_{i} + y_{i+1} - y_{i}$.
\end{definition}
The write sum is a powerful notion because lower bounds on total write sums can give upper bounds on the number of writes. For $M=8$ it has been shown~\cite{Yuval} that without balancing constraints, write sum of $3$ is both sufficient and necessary for every write. The following lemma is key to our upper bound, because it shows cases where a write sum of $4$ is necessary.
\begin{lemma} \label{lm:1}
For any code $\cC_{\ensuremath{3}\mhyphen imb}\left(2,q,t,8\right)$, if a write starts from state $\left(x,y\right)$ satisfying $\left|y-x \right|=d=3$, then the write region of $\left(x,y\right)$ must contain at least two states of write sum $4$ (or higher).
\end{lemma}
\begin{proof}
Let us assume w.l.o.g that the start state is state S showing on Fig.~\ref{fig:Lemma1}. All $5$ states marked as X have write sum of $3$ or less. The write sum of states A, B, and C is $4$. Therefore, in order for the write region to have area at least $M=8$, it must include at least two additional states out of A, B, C, or some other state with higher write sum.
\begin{figure}[htbp]
   \centering
   \includegraphics[height=2.1in,keepaspectratio=true,width=0.5\textwidth]{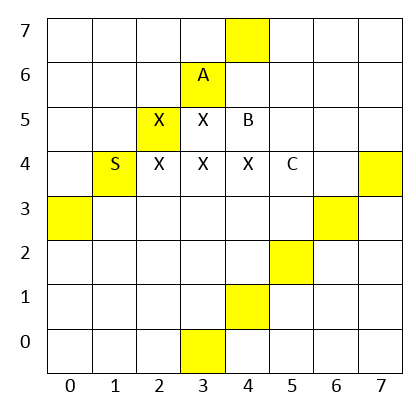}
   \caption{Proof of Lemma~\ref{lm:1} -- only $5$ states (marked by X) have write sum of $3$ or less. States not on or between the shaded diagonals are forbidden due to imbalance greater than $d=3$. }
   \label{fig:Lemma1}
\end{figure}
\end{proof}
The next two lemmas show how any $3$-imbalance code must get to the ``problematic'' state S of Fig.~\ref{fig:Lemma1}.
\begin{lemma} \label{lm:2}
For any code $\cC_{\ensuremath{3}\mhyphen imb}\left(2,q,t,8\right)$, if a write starts from state $\left(x,y\right)$ satisfying $\left|y-x \right|=d-1=2$, then the write region of $\left(x,y\right)$ must contain at least three states of write sum $3$ (or higher), at least one of which has $\left|y'-x' \right|=d=3$ or write sum at least $4$.
\end{lemma}
\begin{proof}
Let us assume w.l.o.g that the start state is state S showing on Fig.~\ref{fig:Lemma2}. All $5$ states marked as X have write sum of $2$ or less. The write sum of states A, B, and C is $3$. Therefore, in order for the write region to have area at least $M=8$ with write sum $3$, it must include the states A, B, C, and state A has $\left|y'-x' \right|=d=3$. If A is not included, then a state with write sum $4$ is required.
\begin{figure}[htbp]
   \centering
   \includegraphics[height=2.1in,keepaspectratio=true,width=0.5\textwidth]{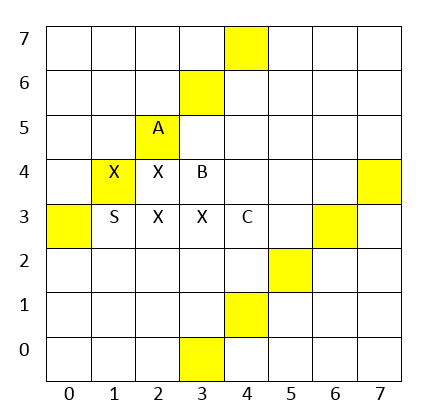}
   \caption{Proof of Lemma~\ref{lm:2} -- state A is required for write sum of $3$ or less.}
   \label{fig:Lemma2}
\end{figure}
\end{proof}
\begin{lemma} \label{lm:3}
For any code $\cC_{\ensuremath{3}\mhyphen imb}\left(2,q,t,8\right)$, if a write starts from state $\left(x,y\right)$ satisfying $\left|y-x \right|=d-2=1$, then the write region of $\left(x,y\right)$ must contain at least two states of write sum $3$ (or higher), at least one of which has $\left|y'-x' \right|=d-1=2$ or write sum at least $4$.
\end{lemma}
\begin{proof}
Let us assume w.l.o.g that the start state is state S showing on Fig.~\ref{fig:Lemma3}. All $5$ states marked as X have write sum of $2$ or less. The write sum of states A, B, and C is $3$. Therefore, in order for the write region to have area at least $M=8$ with write sum $3$, it must include two states out of A, B, C, and both A,C have $\left|y'-x' \right|=d-1=2$. If neither of A,C is included, then a state with write sum $4$ is required.
\begin{figure}[htbp]
   \centering
   \includegraphics[height=2.1in,keepaspectratio=true,width=0.5\textwidth]{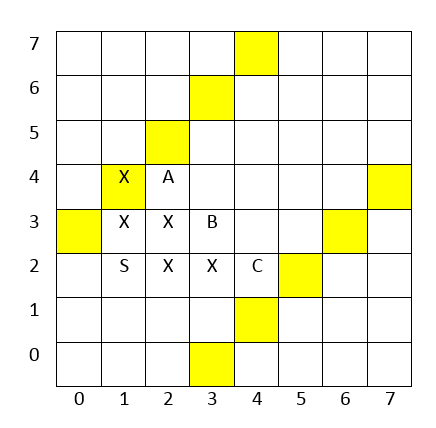}
   \caption{Proof of Lemma~\ref{lm:3} -- A or C are required for write sum of $3$ or less.}
   \label{fig:Lemma3}
\end{figure}
\end{proof}

We are now ready to state the upper bound.
\begin{theorem} \label{th:wiBound}
Given a ${d=3}$-imbalance WOM code $\cC_{\ensuremath{3}\mhyphen imb}\left(2,q,t,8\right)$, the number of guaranteed writes is upper bounded by
\begin{equation} \label{eq:wiBound}
t \leq \left\lfloor \frac{3\left(q-1\right)}{5} \right\rfloor.
\end{equation}
\end{theorem}
\begin{proof}
Throughout the proof we use lower bounds on write sums, but for convenience we omit the term "at least" when we state the value of the write sums. If write sums are strictly larger than the values quoted below, the proof is still correct, and we may even reach the desired lower bound on total write sum earlier by skipping lemmas that assume the lower quoted value. By a simple area argument, after the first write any $M=8$ WOM code needs to use a state with write sum $3$. Next we show that we can invoke Lemmas~\ref{lm:3},~\ref{lm:2},~\ref{lm:1} in sequence to get lower bounds on the write sums of the second, third, and forth write, respectively. Because all states with sum $3$ have $|y-x|$ at least $1$, the conditions of Lemma~\ref{lm:3} are satisfied after the first write (see state S in Fig.~\ref{fig:Lemma3}). By Lemma~\ref{lm:3}, after the second write any code needs a state that satisfies the conditions of Lemma~\ref{lm:2} (see state A in Fig.~\ref{fig:Lemma3} and state S in Fig.~\ref{fig:Lemma2}). By Lemma~\ref{lm:2}, after the third write any code needs a state that satisfies the conditions of Lemma~\ref{lm:1} (see state A in Fig.~\ref{fig:Lemma2} and state S in Fig.~\ref{fig:Lemma1}). Altogether we conclude that for the second, third, and fourth writes we need write sums of $3+3+4=10$. From Lemma~\ref{lm:1}, after the fourth write there are {\em two} states of sum $13$, which implies that one of them has $|y-x|$ at least $1$, and we can again invoke Lemmas~\ref{lm:3},~\ref{lm:2},~\ref{lm:1} in sequence for the subsequent three writes requiring additional $3+3+4=10$ write sums. Continuing this argument periodically, for $t=1+3\ell$ writes, $\ell$ integer, we need a total write sum $s\geq 3+10\ell$, and since $s\leq 2(q-1)$, we get $2(q-1)\geq 3+10\ell=3+10(t-1)/3$. Rearranging, we get $t\leq\lfloor 3(q-1)/5+1/10\rfloor = \lfloor 3(q-1)/5\rfloor$ as needed. To complete the proof, we need to show~\eqref{eq:wiBound} for $t=1+3\ell+r$ writes, for the cases $r=1,2$. In this case, the last $r$ writes each requires write sum of $3$, and we get $s\geq 3+10\ell+3r$. Because both Lemmas~\ref{lm:3},~\ref{lm:2} show existence of {\em two} states with write sum $3$, one of these final states has $|y-x|\geq 1$. Thus we can tighten the relation between $s$ and $q$ to $s \leq x+y \leq q-1 + q-2 =2q-3$. Now we write $2q-3\geq 3+10\ell+3r=3+10(t-1)/3-r/3$. Rearranging, we get $t\leq\lfloor 3(q-1)/5-1/5+r/10\rfloor \leq \lfloor 3(q-1)/5\rfloor$, and for the last inequality we used the fact that $r\leq 2$.
\end{proof}
Note that for $a=3$ the number of writes guaranteed by Construction~\ref{const2} is $\lfloor 3(q-1)/5\rfloor$, which is equal to the upper bound~\eqref{eq:wiBound}, hence Construction~\ref{const2} is optimal.
\subsection{Performance comparison}
Table~\ref{tb:result} presents a summary of known results and bounds for $2$-cell, $M=8$ WOM codes~\cite{Yuval}. By examining the table, we can notice that for $q = 8$, and $q = 16$ (which are currently the practical values of $q$ for NVMs), using Construction~\ref{const2} does not compromise the number of writes compared to optimal unconstrained WOM, while it provides a better imbalance $d=3$. Using $d=2$ constructions does compromise the number of writes for all $q$ values, including $q=8$. It was also verified numerically that for $M=8$, $q\leq 16$ and ${d=3}$-imbalance, codes constructed by Construction~\ref{const2} reach the write-count upper bound of unconstrained codes for \emph{every} value of $q$ in this range. \\
\begin{table} [ht]
\caption{Number of writes for $2$-cell, $M=8$ WOM codes}
\centering
\begin{tabular}{|c||c|c|c|}
  \hline
  \shortstack{q \\ $\,$ } & \shortstack{t - Upper bound \\ d-unconstrained} & \shortstack{t - Construction~\ref{const2} \\ $d=3$ }& \shortstack{t - Construction in~\cite{Yuval} \\ $d=2$  }\\ \hline 
  \hline
  8 & 4 & 4 & 3   \\
  16 & 9 & 9 & 7   \\
  20 & 12 & 11 & 9  \\
  32 & 20 & 18 & 15 \\
  \hline
\end{tabular}  \label{tb:result}
\end{table}
Actually, as we can see in the next corollary, WOM codes constructed by Construction~\ref{const2} are good WOM codes even if ignoring the \wi property. In the following we compare Construction~\ref{const2} to the best known $2$-cell construction \emph{for general $a$} from~\cite{Yuval}.
\begin{corollary}
When $q \geq 1+ \left\lceil \frac{\left(a^2-2 \right) \left(3a-4\right)}{a-2} \right\rceil$, a WOM code $\cC_{a\mhyphen imb}\left(2,q,t,a^2-1\right)$, $2<a \in \mathbb{Z}$ from Construction~\ref{const2} guarantees higher number of writes than a WOM code $\cC\left(2,q,\tilde{t},a^2-1\right)$ constructed by Construction 2 in~\cite{Yuval}.
\end{corollary}
\begin{proof}
The number of writes guaranteed by Construction~\ref{const2} is given by Theorem~\ref{th:wi-t}, while the number writes guaranteed by Construction 2 in~\cite{Yuval} is given by
\begin{equation} \label{eq:const2yuval}
t=\left\lfloor \frac{\left(q-1\right)\left(a+1\right)}{a^2-2} \right\rfloor.
\end{equation}
So, we look for the lowest value of $q$ which guarantees strictly higher number of writes for the $a$-imbalance code. Due to the \emph{floor} function applied on the number of writes, we demand that
\begin{equation}
\frac{3q-3}{3a-4} \geq \frac{\left(q-1\right)\left(a+1\right)}{a^2-2} +1.
\end{equation}
This inequality holds for $q \geq 1+\left(a^2-2\right)\left(3a-4\right)/\left(a-2\right)$. Ceiling this expression ends the proof.
\end{proof}
In addition to offering strictly more writes for these $q$ values, the codes $\cC_{a\mhyphen imb}\left(2,q,t,a^2-1\right)$ have at least as many writes as $\cC\left(2,q,\tilde{t},a^2-1\right)$ for \emph{all values} of $q$. This makes them the best known (unconstrained) WOM codes for these parameters and general $a$ (for $a=3$~\cite{Yuval} has better codes than $\cC\left(2,q,\tilde{t},a^2-1\right)$.)

We can also notice that the number of guaranteed writes offered by Construction~\ref{const2} can reach the unconstrained upper bound for some other values of $M$ and $q$. Table~\ref{tb:result2} presents such pairs of $M$, $q$ values (the $q$ values are taken in the practical range $8 \leq q \leq 16$).
 \begin{table} [ht]
\caption{$M$, $q$ values for which Construction~\ref{const2} attains the (unconstrained) upper bound. }
\centering
\begin{tabular}{|c||c|}
  \hline
  M & values of q  \\ \hline 
  \hline
  15 & 9   \\  \hline
  24 & 9,10,12,13,16   \\ \hline
  35 & 11,15 \\ \hline
  48 & 13,14 \\
  \hline
\end{tabular}  \label{tb:result2}
\end{table}

\section{Lattice-Based \wi WOM Codes}\label{sec:EWLbhatia}
To present \wi codes in the lattice approach we start with some formal definitions.
\begin{definition}
A variable-rate \textbf{WOM code} $\cC\left(\cellnum,q,t,\vect{M}\right)$ is a code applied to a size $\cellnum$ block of $q$-ary cells, and guaranteeing $t$ writes, where the input size for the $i$-th write is $M_i$ taken from the vector $\vect{M}=\left(M_1, \ldots ,M_t\right)$.
\end{definition}
\begin{definition}
The \textbf{sum-rate} $\cR_{sum}$ of a WOM code $\cC\left(n,q,t,\vect{M}\right)$ is defined as
\begin{equation} \label{eq:sumrate}
\cR_{sum} = \frac{\sum_{i=1}^{t}\log_2\left(M_i\right)}{\cellnum}.
\end{equation}
In other words, the sum rate is the total number of written bits divided by the number of memory cells.
\end{definition}
\subsection{Background and review of known results}
Lattice-based WOM codes were first proposed in~\cite{s4} by Kurkoski, and were further extended by Bhatia et al. in~\cite{Bhatia},\cite{bhatia2}. In the lattice approach, the $n$-dimensional discrete space of physical states $\{0,\ldots,q-1\}^n$ is approximated as the continuous space $[0,q-1]^n$. In that approximation the $i$-th write's input size $M_i$ is approximated by an area $\iarea$ in the continuous space. Given the number of writes $t$, the space $[0,q-1]^n$ is partitioned to $t$ disjoint regions, each allocated to a write in the sequence of $t$ writes. In the process of paritioning the space, each write $i$ is allocated an area $\iarea$. The objective of the partition is to maximize the product of the areas $\prod_{i=1}^{t}\iarea$, because this would approximate maximizing the sum-rate of~\eqref{eq:sumrate}. After the continuous space is partitioned, a discretization algorithm assigns labels to discrete physical states in every region to obtain the WOM decoding function. The advantage of the lattice approach over the direct construction approach of~\cite{Yuval} and Section~\ref{sec:optimal_cnst} is that it can use analytic geometry to find region partitions with good properties. The key disadvantage is that optimality can only be guaranteed for the continuous approximation of the space, while the direct approach yields explicit optimal codes in the true discrete space.

In a nutshell, a $2$-cell $2$-write lattice-based WOM code is constructed by partitioning the $2$-dimensional space $[0,q-1]^2$ into $2$ regions, one for each write. The first write gets allocated an area of $\farea$ confined between the $x$,$y$ axes and the boundary curve (see Fig.~\ref{fig:fig_for_lemma2}). This leaves the second write an area $\sarea$ of a rectangle confined between the boundary curve and the $x=q-1$, $y=q-1$ lines. The boundary curve is chosen to maximize $\farea\cdot \sarea$. It was shown~\cite{s4} that the optimal boundary takes the shape of a rectangular {\em hyperbola}. The constructions of lattice-based WOM codes were generalized to any number of writes $t$~\cite{bhatia2}, and (non explicitly) to any number of cells $n$~\cite{Bhatia}. In~\cite{Bhatia} the lattice approach is applied to both variable-rate and fixed-rate WOM codes. In order to make this paper cohasive, we stick with the notations of Section~\ref{sec:optimal_cnst} (rather than those of the original papers~\cite{Bhatia},\cite{bhatia2}) with one exception: we replace the discrete cardinalities $M_i$ of the input sizes with continuous cardinalities $\iarea$. With taking measures to avoid confusion, we slightly abuse the term sum-rate to describe the continuous areas $\iarea$ in lieu of the discrete input sizes $M_i$.

The following is a restatement of a result from~\cite{bhatia2}.
\begin{theorem}~\cite{bhatia2} \label{th:Bhatia}
The optimal continuous boundary between the writes of a $2$-cell $2$-write lattice-based WOM code is given by the following equation of a hyperbola
\begin{equation} \label{eq:bound}
\beta\left(x\right) = q-1 - \frac{\omega_2 \left(q-1\right)^2}{q-1-x},
\end{equation}
where $x \in \left[0,\left(q-1\right)\left(1-\omega_2\right)\right]$, and
\begin{equation}
\omega_2 = -\frac{1}{2}\left[W_{-1} \left(\frac{-1}{2 \sqrt{e}}\right) \right]^{-1}.
\end{equation}
$W_{-1}$ is the real branch of the Lambert $W$ function~\cite{LambertW} satisfying $W\left(x\right)<-1$.
\end{theorem}

%
%
In the $xy$ plane the optimal boundary of Theorem~\ref{th:Bhatia} is given in closed form as the curve $y=\beta\left(x\right)$. This optimal boundary was derived as follows. Let us assume the optimal boundary between the two writes is given by some function $y=\betag \left(x\right)$. $\farea$ is the area under $\betag\left(x\right)$ while $\sarea$ is calculated as the area of a rectangle formed by some point on $ \betag\left(x\right)$ with the $x=q-1$ and $y=q-1$ lines (see Fig.~\ref{fig:fig_for_lemma2}).
\begin{figure}[htbp]
   \centering
   \includegraphics[height=2.1in,keepaspectratio=true,width=0.5\textwidth]{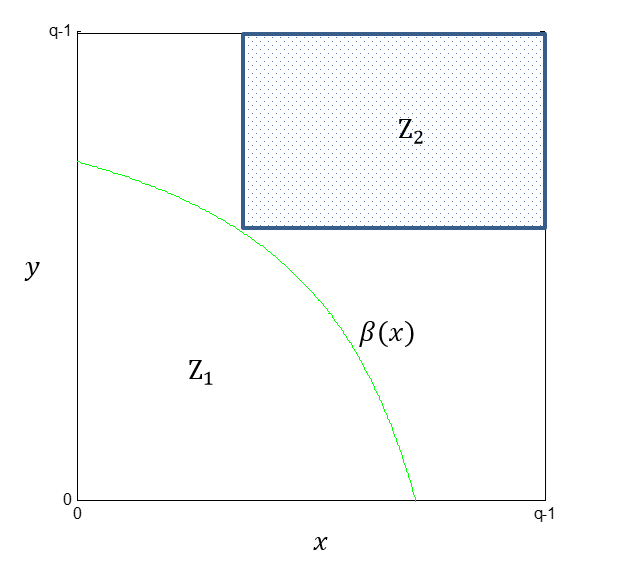}
   \caption{The optimal boundary $\beta\left(x\right)$ of a $2$-cell $2$-write lattice-based WOM code. The $x$ and $y$ axes represent the memory cells $c_1$ and $c_2$ respectively. The cardinality of the first write $\farea$ is the area under $\beta\left(x\right)$. The cardinality of the second write $\sarea$ is the area of the dotted rectangle.}
   \label{fig:fig_for_lemma2}
\end{figure}
In order to get the optimal sum-rate code, we need to find the boundary $\beta\left(x\right)$ that brings $\farea \cdot \sarea$ to a maximum. Given that each point on $\betag\left(x\right)$ can span a rectangle with different area, $\sarea$ is the minimum of the areas of all possible rectangles. Therefore, the following optimization problem~\cite{bhatia2} is solved
\begin{equation} \label{eq:sumrate_crit}
\max_{\betag\left(x\right)}\left\{\int\limits_0^{x_{sup}} \betag\left(x\right) dx  \cdot \min_{\forall x,~ \betag\left(x\right)} \left[ \left(q-1-x\right) \left(q-1-\betag\left(x\right)\right) \right] \right\},
\end{equation}
where $x_{sup}$ is the maximal value of the support of $\betag\left(x\right)$. The solution of~\eqref{eq:sumrate_crit} gives the hyperbola boundary $\beta\left(x\right)$ of Theorem~\ref{th:Bhatia}.
\\After deriving the continuous boundaries between the writes, the WOM code is constructed by discretization and a label assignment algorithm.
\begin{example} \label{ex2}
Let us consider the following WOM code~\cite{Bhatia} $\cC\left(n=2,q=8,t=2,\vect{M}=\left(24,23\right)\right)$. This code is applied on a pair of $8$-level memory cells, enabling $2$ guaranteed writes of input sizes $24$ and $23$ for the first and second writes, respectively. The decoding function of this code is presented in Fig.~\ref{fig:Example2}.
\begin{figure}[htbp]
   \centering
   \includegraphics[height=2.1in,keepaspectratio=true,width=0.5\textwidth]{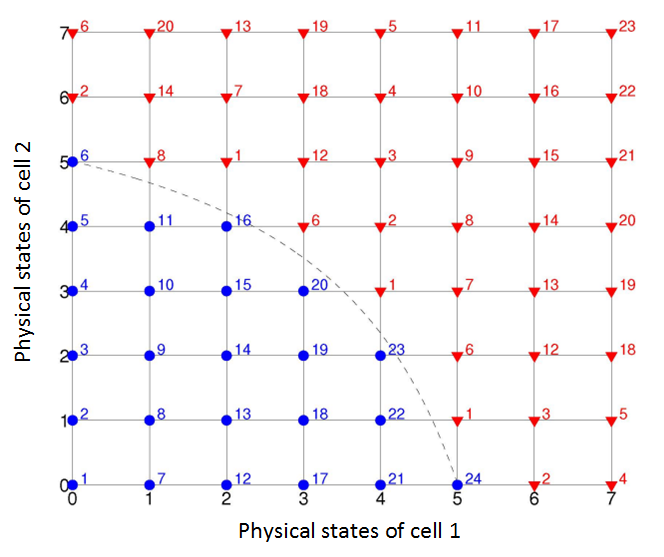}
   \caption{\cite{Bhatia} Decoding function $\psi$ for the code $\cC\left(2,8,2,(M_1,M_2)=\left(24,23\right)\right)$. Logical states for the first write are the labels of the circle physical states, and for the second write are the labels of the triangle physical states. The boundary between the two writes is an hyperbola. }
   \label{fig:Example2}
\end{figure}
The boundary between the two writes is the following hyperbola:
\begin{equation}
\beta\left(x\right) = 7 - \frac{13.948}{7-x}.
\end{equation}
\end{example}
\subsection{Construction for lattice-based \wi WOM codes}
In this sub-section we present a construction of \wi WOM codes after applying the imbalance model (Definition~\ref{def:imbalance}) to the continuous approximation of the lattice approach. Our main result toward that is a closed-form characterization of the optimal boundary for $2$-cell $2$-write \wi WOM codes.
\begin{theorem} \label{th:main_result}
When $d \leq \frac{3}{7}\left(q-1\right)$, the optimal boundary for a maximal sum-rate \wi lattice-based WOM code $\cC_{d\mhyphen imb}\left(2,q,2,(\farea,\sarea)\right)$ is given by $\beta_d\left(x\right)$, satisfying
\begin{gather} \label{eq:parabola2}
\left(q-1-x\right)\left(q-1-\beta_d\left(x\right)\right)-\\ \nonumber
\frac{\left(q-1-d-x\right)^2}{2}-\frac{\left(q-1-d-\beta_d\left(x\right)\right)^2}{2}=d\left(q-1\right) -\frac{5d^2}{6},
\end{gather}
and the optimal sum-rate is given by
\begin{equation}
\cR_{sum} = \log_2\left[d\left(q-1\right) -\frac{5d^2}{6}\right].\label{eq:parabola_sumrate}
\end{equation}
\end{theorem}
It can be checked that~\eqref{eq:parabola2} implies that the curve $y=\beta_d\left(x\right)$ is a {\em parabola}.

Before proving this Theorem we present the following lemmas. The first lemma finds the shape of the boundary $\beta_d\left(x\right)$ that yields for the second write identical areas among the points on the boundary. Note that with the \wi constraint the areas $\sarea$ are bounded by the lines $y=x\pm d$, in addition to the bounding by the lines $x=q-1$ and $y=q-1$ in the unconstrained case (see Fig.~\ref{fig:fig_for_lemma}.)
\begin{lemma} \label{lm:opt_b}
Given a \wi lattice-based WOM code $\cC_{d\mhyphen imb}\left(2,q,2,(\farea,\sarea)\right)$ with $d\leq\sqrt{\frac{2}{3}\sarea}$, the boundary $\beta_d\left(x\right)$ that yields identical $\sarea$ values for all points on $\beta_d\left(x\right)$ is given by
\begin{gather} \label{eq:parabola}
\left(q-1-x\right)\left(q-1-\beta_d\left(x\right)\right)-\\ \nonumber
\frac{\left(q-1-d-x\right)^2}{2}-\frac{\left(q-1-d-\beta_d\left(x\right)\right)^2}{2}=\sarea.
\end{gather}
\end{lemma}

\begin{proof}
The constraint induced by the \wi model is that all the valid physical states are bound by the lines $y=x\pm d$, as can be seen in Fig.~\ref{fig:fig_for_lemma}.
Therefore, $\beta_d\left(x\right)$ is a function on which every point spans an equal-area shape with the $x=q-1$, $y=q-1$ and $y=x\pm d$ lines. We now turn into calculating this area. First we have the area of a rectangle (denoted by dashed lines in Fig.~\ref{fig:fig_for_lemma}) given by $\left(q-1-x\right)\left(q-1-\beta_d\left(x\right)\right)$.
\begin{figure}[htbp]
   \centering
   \includegraphics[height=2.1in,keepaspectratio=true,width=0.5\textwidth]{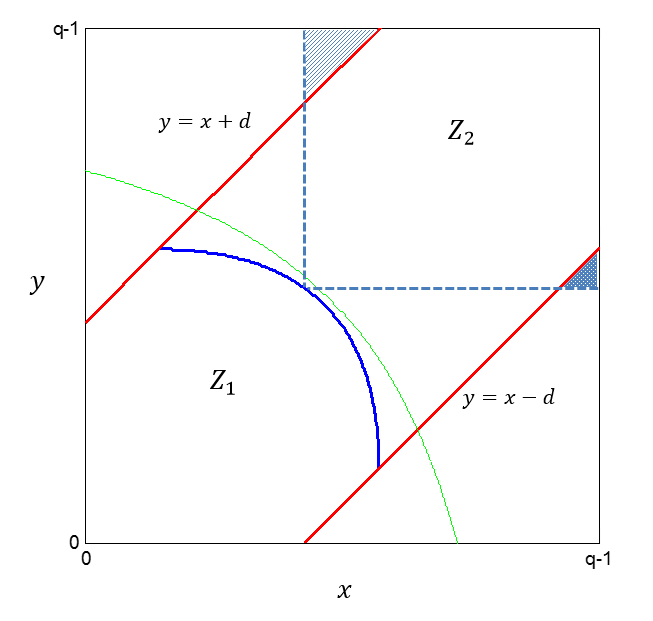}
   \caption{Proof of Lemma~\ref{lm:opt_b}: the boundary inside the $y=x\pm d$ lines is the parabola~\eqref{eq:parabola}, the other boundary is the rectangular hyperbola stated in~\eqref{eq:bound}. The two triangles are the area that should be subtracted from the rectangle (with dashed borders) in order to calculate $\sarea$. The $x$ and $y$ axes represent the memory cells $c_1$ and $c_2$ respectively.}
   \label{fig:fig_for_lemma}
\end{figure}
From this rectangle we need to subtract the area of the two triangles. It is easy to verify that the area of the top triangle is given by $\left(q-1-d-x\right)^2/2$ and the area of the bottom triangle is given by $\left(q-1-d-\beta_d\left(x\right)\right)^2/2$. Subtracting these areas and equating to the cardinality of the second write $\sarea$ yields~\eqref{eq:parabola}. Note that for the two subtracted triangles to exist, the $x$ coordinate of the intersection point between $\beta_d\left(x\right)$ and $y=x-d$ must be at most $q-1-d$. Given that the intersection point is $\left(q-1-\frac{\sarea}{2d}-\frac{d}{4},q-1-\frac{\sarea}{2d}-\frac{5d}{4}\right)$, we get that in order for the triangles to exist we need
\begin{equation}
q-1-\frac{\sarea}{2d}-\frac{d}{4} \leq q-1-d.\label{eq:triangle_cond}
\end{equation}
After some manipulations we get that the condition~\eqref{eq:triangle_cond} is equivalent to $d \leq \sqrt{\frac{2}{3}\sarea}$ given in the Lemma statement.
\end{proof}
The second lemma finds the relation between the areas of the first and second writes for boundaries $\beta_d\left(x\right)$ in the form given in Lemma~\ref{lm:opt_b}.
\begin{lemma} \label{lm:card}
For a $\beta_d\left(x\right)$ given in~\eqref{eq:parabola}, the cardinality of the first write is given by
\begin{equation}
\farea = 2d\left(q-1\right) -\sarea -\frac{5d^2}{3}.
\end{equation}
\end{lemma}
\begin{proof}
The cardinality of the first write is the area bound by the $x$ and $y$ axes, the $y=x\pm d$ lines, and $\beta_d\left(x\right)$ from~\eqref{eq:parabola}. Due to the $xy$ symmetry of the problem, we first calculate the area between $\beta_d\left(x\right)$ and the line $y=x$. Then we subtract the area between $\beta_d\left(x\right)$ and $y=x+d$, and finally we multiply the outcome by $2$. In order to do so, we first calculate the intersection points between $\beta_d\left(x\right)$ and the lines $y=x+d$ and $y=x$. It is easy to verify that the intersection points are $\left(q-1-\frac{\sarea}{2d}-\frac{5d}{4},q-1-\frac{\sarea}{2d}-\frac{d}{4}\right)$ and $\left(q-1-\frac{\sarea}{2d}-\frac{d}{2},q-1-\frac{\sarea}{2d}-\frac{d}{2}\right)$, respectively. Therefore, the desired area $\farea$ can be calculated by
\begin{gather} \label{eq:intS}
\frac{\farea}{2} = \int\limits_0^{q-1-\frac{\sarea}{2d}-\frac{d}{2}} \left[ \beta_d\left(x\right) -x \right]dx - \int\limits_0^{q-1-\frac{\sarea}{2d}-\frac{5d}{4}}  \left[ \beta_d\left(x\right) - \left(x+d\right) \right]dx \\ \nonumber
= \int\limits_{q-1-\frac{\sarea}{2d}-\frac{5d}{4}}^{q-1-\frac{\sarea}{2d}-\frac{d}{2}} \left[ \beta_d\left(x\right) -x \right]dx + \int\limits_0^{q-1-\frac{\sarea}{2d}-\frac{5d}{4}} d \cdot dx.
\end{gather}
>From Lemma~\ref{lm:opt_b} it is not hard to see that for $x \leq q-1-\frac{\sarea}{2d}-\frac{d}{4} $, $\beta_d\left(x\right)$ is given by
\begin{equation}
\beta_d\left(x\right) = \sqrt{4\left(q-1\right)d-d^2-2\sarea-4dx}-d+x.
\end{equation}
Substituting this $\beta_d\left(x\right)$ in~\eqref{eq:intS} gives the expression
\begin{gather}
\frac{\farea}{2} = -\frac{1}{6d}\left(4\left(q-1\right)d-d^2-2\sarea-4dx\right)^{\frac{3}{2}}\Big|_{q-1-\frac{\sarea}{2d}-\frac{5d}{4}}^{q-1-\frac{\sarea}{2d}-\frac{d}{2}} \\ \nonumber
-\frac{3d^2}{4} + d\left(q-1-\frac{\sarea}{2d}-\frac{5d}{4}\right) = d\left(q-1\right) -\frac{\sarea}{2}-\frac{5d^2}{6}.
\end{gather}
\end{proof}
With the help of Lemmas~\ref{lm:opt_b} and~\ref{lm:card}, we can now prove Theorem~\ref{th:main_result}.\\
\begin{proof}
In order to find the maximal sum-rate of $\cC_{d\mhyphen imb}\left(2,q,2,(\farea,\sarea)\right)$ with imbalance parameter $d$, we now need to adjust the optimization problem of~\eqref{eq:sumrate_crit} to the \wi model, and to find the values of $\farea$ and $\sarea$ that maximize $\farea\cdot \sarea$. By a similar argument to the one proved in~\cite{bhatia2}, the minimization in~\eqref{eq:sumrate_crit} implies that the optimal sum-rate boundary must satisfy the identical-area condition of Lemma~\ref{lm:opt_b}. Then among the $\beta_d\left(x\right)$ curves of Lemma~\ref{lm:opt_b} parametrized by $\sarea$, the corresponding value of $\farea$ is determined by Lemma~\ref{lm:card}. This implies that
\begin{equation}
\farea\cdot \sarea = \left[ 2d\left(q-1\right) -\sarea -\frac{5d^2}{3} \right] \sarea.
\end{equation}
Taking the derivative of the right-hand side with respect to $\sarea$ and equating to $0$ gives
\begin{equation}
\sarea = d\left(q-1\right) -\frac{5d^2}{6} = \farea.
\end{equation}
This proves the right-hand sides of~\eqref{eq:parabola2} and~\eqref{eq:parabola_sumrate}.
\end{proof}
There are two interesting conclusions from Theorem~\ref{th:main_result}. First is that the introduction of the \wi constraint changed the shape of the curve from a hyperbola to another regular shape: a parabola. Second is that for the \wi case the cardinalities that maximize the sum-rate turn out to be the fixed-rate cardinalities. This favorable property does not exist in the unconstrained case (for unconstrained WOM the fixed-rate property costs sub-optimality in sum-rate). We next show an example of a code constructed with the help of Theorem~\ref{th:main_result} followed by the discretization step.
\begin{example}\label{ex_discrete_d}
Let us consider the following $8$-level $2$-cell $2$-write \wi WOM code with imbalance parameter of $d=3$, $\cC_{3\mhyphen imb}\left(2,8,2,(\farea,\sarea)\right)$. By Theorem~\ref{th:main_result}, the optimal boundary is given by
\begin{equation}
2\left(7-x\right)\left(7-\beta_3\left(x\right)\right) - \left(4-x\right)^2 - \left(4-\beta_3\left(x\right)\right)^2 = 27.\label{eq:ex_d_lattice}
\end{equation}
The continuous cardinalities corresponding to the boundary of~\eqref{eq:ex_d_lattice} are $(\farea,\sarea)=(13.5,13.5)$. But after discretization we obtain in Fig.~\ref{fig:Example4} discrete cardinalities $(M_1,M_2)=(18,21)$ (it is possible to have $M_i>\iarea$ because a point can be in the region without its entire unit square). In particular, the resulting WOM code is not fixed-rate even though the continuous boundary is fixed-area. The decoding function of $\cC_{3\mhyphen imb}\left(2,8,2,(M_1,M_2)=(18,21)\right)$ is presented in Fig.~\ref{fig:Example4}.
\begin{figure}[htbp]
   \centering
   \includegraphics[height=2.1in,keepaspectratio=true,width=0.5\textwidth]{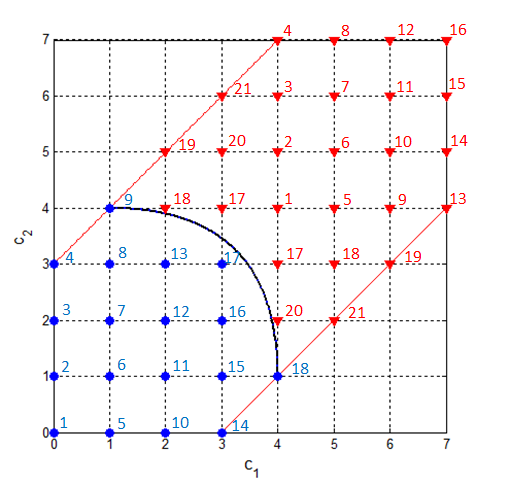}
   \caption{Decoding function $\psi$ for the code $\cC_{3\mhyphen imb}\left(2,8,2,(M_1,M_2)=(18,21)\right)$. Logical states for the first write are the labels of the circle physical states, and for the second write are the labels of the triangle physical states. The boundary between the two writes is a parabola. }
   \label{fig:Example4}
\end{figure}
The sum-rate of the non-balanced code $\cC\left(2,8,2,(M_1,M_2)=(23,24)\right)$ of Example~\ref{ex2} (with the same $q$, $n$ and $t$) is $4.55$. By constraining the code with the $d=3$ imbalance parameter the sum-rate in Example~\ref{ex_discrete_d} is reduced to $4.28$ ($6\%$ reduction). Let us now compare the lattice-based results to direct \wi WOM codes constructed by Construction~\ref{const2}. To obtain a $t=2$ code from Construction~\ref{const2} with $q=8$, the largest rate parameter is $a=4$, corresponding to $M=15$. This gives a code $\cC_{3\mhyphen imb}\left(2,8,2,(M_1,M_2)=(15,15)\right)$, which has a sum-rate of $3.91$, lower than the lattice-based construction for this example.
\end{example}
\subsection{Performance Comparison}
We now compare the unconstrained lattice-based WOM codes of~\cite{Bhatia},\cite{bhatia2} with the new lattice-based \wi WOM codes. We make the comparison over the {\em continuous} sum-rate calculated from $\farea$,$\sarea$ (before discretization). Table~\ref{tb:comp2} presents the sum-rate of the codes for different values of $d$ and $q$. The reader can notice that when the imbalance takes the highest value allowed by Theorem~\ref{th:main_result}: $d=\left\lfloor \frac{3\left(q-1\right)}{7}\right\rfloor$, the sum-rate is compromised by approximately $5\%$. Naturally, when $d$ decreases the sum-rate decreases due to the more limiting constraint imposed by the $d$-imbalance model.
\begin{table} [ht]
\caption{Continuous sum-rate comparison of $2$-cell lattice-based WOM codes for $q=8$ and $q=16$. }
\centering
{\extrarowsep=1mm
\begin{tabu}{|c||c|c|}
  \hline
  q & $ d $ & \textbf{$\cR_{sum}$}   \\\tabucline[1.2pt]{-}
  8 & - &  3.97 \\  \hline
  8 & 3 &  3.75  \\ \hline
  8 & 2 & 3.42 \\\tabucline[1pt]{-}
  16 & - &  6.17 \\  \hline
  16 & 6 &  5.91  \\ \hline
  16 & 5 &  5.76 \\ \hline
  16 & 4 &  5.54 \\ \hline
  16 & 3 &  5.29 \\ \hline
\end{tabu} } \label{tb:comp2}
\end{table}

\section{Wordline ICI Reduction by \wi WOM Codes}\label{sec:practical}
To this point, we have presented code constructions that bound the imbalance between \emph{two} memory cells. In flash practice, many more than two cells are updated together in a memory {\em page}, also called a {\em wordline}. To match the write granularity of the flash architecture, we will use a code $\cC_{\wiAB}\left(2,q,t,M\right)$ {\em on each pair} of adjacent cells in a wordline. Hence, a wordline includes concatenated pairs of WOM-coded cells, as depicted in Fig.~\ref{fig:block}.
\begin{figure}[htbp]
   \centering
   \includegraphics[height=2.1in,keepaspectratio=true,width=0.5\textwidth]{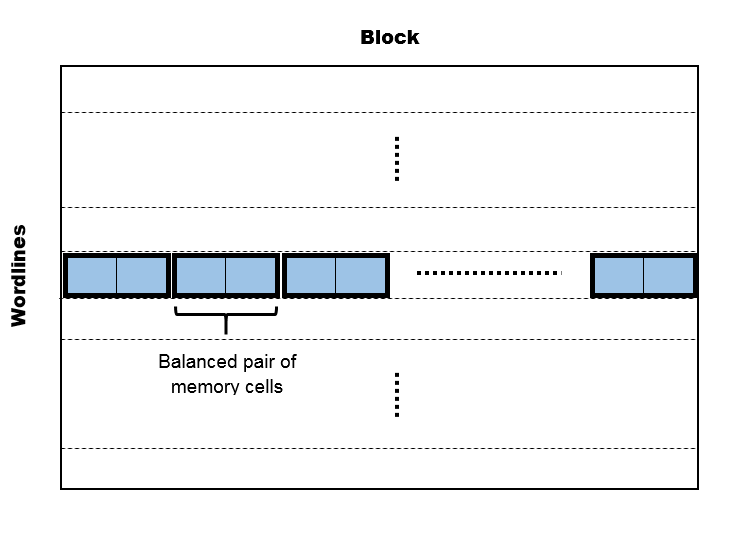}
   \caption{ A memory block consisting of wordlines, where each wordline includes concatenated pairs of $d$-imbalance codewords.}
   \label{fig:block}
\end{figure}
In the sequel we say that a set of cell levels is {\em $d$-balanced} if they satisfy the \wi constraint of~\eqref{eq:wi-d}.
\subsection{Inter-codeword balancing}
As can be seen in the following example, simple concatenation of codewords (of two memory cells each) does not guarantee that the $d$-imbalance property is maintained between any pair of adjacent cells.
\begin{example}\label{ex:trivial_update}
Let us assume we code two adjacent pairs of cells in the same wordline by the \wi WOM code $\cC_{3\mhyphen imb}\left(2,q,t,8\right)$ presented in Fig.~\ref{fig:WIcodeExp}. Suppose the first pair holds the logical state of $'1'$ by the physical state  $\left(1,0\right)$, and the second pair holds the logical state of $'5'$ by the physical state $\left(2,1\right)$. To this end, all four cells $\left(1,0,2,1\right)$ are $3$-balanced. Now, suppose that in the next wordline update we wish to only update the logical value of the second pair from $'5'$ to $'2'$. Therefore, the physical state of the first pair remains at $\left(1,0\right)$ while the physical state of the first pair is updated to $\left(4,2\right)$. Therefore, the four cell levels $\left(1,0,4,2\right)$ are no longer $3$-balanced because the two center cell levels $\left(0,4\right)$ are $4$ apart.
\end{example}
In order to keep the $d$-imbalance property among all the cells in the wordline, it is possible to use WOM codes with the \emph{synchronous} property~\cite{syncWOM}, which means that each of the $t$ writes has a disjoint set of physical states. This way, it is guaranteed that none of the codewords remain in the same physical state through the write sequence, and balance is maintained across codewords as well. However, requiring the synchronous property compromises the sum-rate relative to non-synchronous WOM codes with the same code parameters. \\
Interestingly, as we show next, we can maintain the $d$-imbalance property while using the (non-synchronous) Construction~\ref{const2} with a more clever wordline update process. We start with some formal definitions and a proposition.
\begin{definition}
Given a code $\cC_{\wiAB}\left(2,q,t,M\right)$ applied to a wordline with $2N$ cells, $\cM^i$ is defined as the \textbf{wordline data vector} of the $i$-th write, where the elements of $\cM^i = \left[m_1^i,m_2^i,\ldots,m_N^i\right]$ are the logical states of the $N$ codewords. We also define $\vect{c^i}$ as the \textbf{physical states vector} of the $i$-th write, where the elements of $\vect{c^i}=\left[\underline{c}_1^i,\underline{c}_2^i,\ldots,\underline{c}_{N}^i\right]$ are the physical states of the $N$ codewords.
\end{definition}
\begin{definition}
A physical state $\left(c_1,c_2\right)$ is called a \textbf{frontier} state of the $i$-th write if it can be reached after $i$ writes, and no other state $(c'_1,c'_2)$ with $c'_1\geq c_1$, $c'_2\geq c_2$ can be reached after $i$ writes.
\end{definition}
Informally, frontier states are the ``worst'' states to be in after $i$ writes, and the only ones we need to consider for the code correctness.
\begin{definition}
The frontier states of the $i$-th write are denoted by the set $\cF_i$. The subset of frontier states that can be accessed from the physical state $\left(c_1,c_2\right)$ is denoted by $\cF_i\left(c_1,c_2\right)$.
\end{definition}
\begin{example} \label{ex:non_balanced}
Let us consider the \wi WOM code $\cC_{3\mhyphen imb}\left(2,q,t,8\right)$ whose decoding function is given in Fig.~\ref{fig:WIcodeExp}. The frontier states of the first write are $\cF_1 = \left\{\left(1,2\right), \left(2,1\right) \right\}$. The frontiers of the second and third writes are given by $\cF_2 = \left\{\left(2,4\right), \left(3,3\right), \left(4,2\right) \right\}$ and $\cF_3 = \left\{\left(5,5\right) \right\}$, respectively.
\end{example}
Our objective now is to present an update process that guarantees that after each write all the cells in the wordline are $d$-balanced. In particular, every pair of adjacent cells -- both within and across codewords -- will be $d$-balanced. The following proposition provides the basis for that update process.
\begin{proposition} \label{prop:frontiers}
Let us consider a \wi WOM code $\cC_{\wiAB}\left(2,q,t,M\right)$ constructed by Construction~\ref{const2}. If $(c_1,c_2)$ is a frontier state of the $(i-1)$-th write, and $(c_3,c_4)$ is a frontier state of the $i$-th write, then all of $c_1,c_2,c_3,c_4$ satisfy the $d=a$-imbalance constraint.
\end{proposition}
\begin{proof}
Due to the periodic nature of $\cC_{\wiAB}\left(2,q,t,M\right)$, it is sufficient to prove the statement for the first three writes only. From Fig.~\ref{fig:WIcode} we can see that the frontiers of the first write are $\cF_1=\left\{ \left(a-1,a-2\right),\, \left(a-2,a-1\right) \right\}$. The frontiers of the second and third writes are given by $\cF_2=\left\{ \left(2a-2,2a-4\right),\, \left(2a-3,2a-3\right),\, \left(2a-4,2a-2\right)\right\}$ and $\cF_3=\left\{ \left(3a-4,3a-4\right)\right\}$, respectively. It can now be easily verified that any pair of levels taken from $\cF_{i-1}\cup \cF_i$ are at most $d=a$ apart.
\end{proof}
The implication of Proposition~\ref{prop:frontiers} is that it is sufficient to keep all the $N$ physical states in a wordline between frontier states of two adjacent writes $i-1$ and $i$, {\em inclusive} of states of both frontiers. This is achieved by the wordline update process given in Algorithm~\ref{alg:update}.
\begin{algorithm}[ht]\label{alg:update}
\SetInd{.03in}{.03in}
\SetKwInOut{Input}{input}
\SetKwInOut{Output}{output}
\caption{WordlineUpdate}
\Input{$\cM^{i}$,$\vect{c^{i-1}}$}
\Output{$\vect{c^i}$}

$\cM^{i-1} = \psi\left( \vect{c^{i-1}} \right)$

for $j=1$ to $N$

\qquad choose $\underline{c}\in \cF_{i-1}\left(\underline{c}_{j}^{i-1}\right)$ arbitrarily \\
\qquad $\underline{c}_{j}^{i} = \mu\left(\underline{c}, m_j^i \right)$\\

end
\end{algorithm}

In simple words, Algorithm~\ref{alg:update} guarantees that after $i$ writes every physical state in the wordline will be at least in a frontier state of the $(i-1)$-th write (and at most in a frontier state of the $i$-th write). We now prove the correctness of the wordline update process.
\begin{theorem}
If a WOM code $\cC_{\wiAB}\left(2,q,t,M\right)$ by Construction~\ref{const2} is used in a full wordline with the update process of Algorithm~\ref{alg:update}, then the $d$-imbalance property is maintained on all the cells of the wordline.
\end{theorem}
\begin{proof}
We prove that after $i$ writes every physical state is bounded (element-wise) from above by an $i$ frontier state and from below by an $i-1$ frontier state. This claim together with Proposition~\ref{prop:frontiers} would establish the theorem statement. By Algorithm~\ref{alg:update}, the update function $\mu$ is invoked with a physical-state argument from $\cF_{i-1}$. So trivially by the WOM property the output of the update function must be bounded from below by a state in $\cF_{i-1}$. By the properties of the code, for any $m_j^i$ the output of $\mu$ is bounded from below by a state in $\cF_{i}$.
\end{proof}

\begin{example}
Let us now return to Example~\ref{ex:trivial_update}. The example begins with two logical states $'1'$ and $'5'$ stored in two adjacent codewords of $\cC_{3\mhyphen imb}\left(2,q,t,8\right)$. After updating the second logical state from $'5'$ to $'2'$, the physical states $\left(1,0,4,2\right)$ were no longer $3$-balanced. However, if we use Algorithm~\ref{alg:update} for update, the physical state of the first (not updated) logical state becomes $\left(3,2\right)$, which is the nearest physical state representing $'1'$ starting from a frontier state of the first write. The physical states are now $\left(3,2,4,2\right)$, and they satisfy the $3$-imbalance constraint.
\end{example}

\subsection{ICI analysis} \label{subsec:ici_an}
We now analyze the ICI reduction by using a \wi WOM code. We start with some background concerning the noise model and bit-error rate calculations.
\subsubsection{Noise model}
As was previously explained, different memory levels are represented by different levels of electrical charge or voltage. Due to many inherent physical limitations~\cite{errors2},\cite{errors1}, these voltage values include noise. As a consequence, each memory level is distributed over a range of voltages around the desired voltage level. The most common and simple model for flash voltage distributions is the Gaussian model, whereby a cell voltage level is one of $q$ discrete levels plus an additive Gaussian noise  $\nu \thicksim \cN\left(0,\sigma^2\right)$.

During read operation, the memory level is determined by a sequence of comparisons of the cell voltage level to some reference voltage levels~\cite{JSAC}. Therefore, the bit-error rate (BER) is given by~\cite{imm}
\begin{equation} \label{eq:BER}
BER = \frac{2\left(q-1\right)}{q}Q\left(\frac{V_{ref}}{\sigma}\right),
\end{equation}
where $V_{ref}$ is some normalized reference level between two adjacent voltage levels, and the $Q\left(x\right)$ function is
\begin{equation}
Q\left(x\right) = \frac{1}{\sqrt{2\pi}}\int_x^{\infty} e^{-\frac{t^2}{2}}dt.
\end{equation}
\subsubsection{ICI model}
A recent work~\cite{cai} combined a theoretic ICI model with empirical measurements and presented a practical model for ICI. According to this model, the threshold voltage change of some victim cell is given by
\begin{equation} \label{eq:ici_model}
\Delta V_{victim} = \sum\limits_x \sum\limits_y \alpha\left(x,y\right) \Delta V_{neighbor}\left(x,y\right)+\alpha_0 V_{victim}^{before},
\end{equation}
where $\alpha\left(x,y\right)$ and $\alpha_0$ are fitting coefficients, and $V_{victim}^{before}$ is the threshold voltage of the victim cell before interference. It was shown in~\cite{cai}, that practically the $\alpha\left(x,y\right)$ coefficients are significant for up to three closest neighboring cells only. Therefore, a cell is likely to suffer ICI if there is a significant voltage change in one of its closest neighboring cells. \\
However, as was described in section~\ref{sec:def}, the ISPP (incremental step pulse program) write process is also a key feature in the ICI mechanism. In the ISPP method each program level induces a sequence of program pulses followed by a verification process to assure proximity to the target level. Each program step increases the voltage level of a cell by $\Delta V_{pp}$, which is significantly smaller than the actual voltage levels representing memory values. The voltage raise due to a single program step can be modeled~\cite{els_mod} by adding a uniform random variable in the range of $\left[0,\Delta V_{pp}\right]$. The program process that includes a sequence of $\seqN$ program steps is thereby modeled by the sum of $\seqN$ such uniform random variables, giving the well-known Gaussian shaped voltage level distributions, when $\seqN \gg 1$.
\begin{proposition} \label{prop:ici_noise}
The ICI noise $\nu_{ici}$ of a victim cell, due to $\seqN$ program steps in a nearby aggressor cell, is an Irwin-Hall distributed random variable, with mean $\mu_{ici}=\alpha \frac{\seqN \Delta V_{pp}}{2}$ and variance $\sigma_{ici}^2 = \alpha^2 \frac{\seqN}{2}\frac{\left(\Delta V_{pp} +1\right)^2-1}{12}$, where $\alpha$ is the capacitance ratio between the two cells and $\Delta V_{pp}$ is the ISPP voltage step.
\end{proposition}
\begin{proof}
The ICI noise is the sum of $\seqN$ independent random variables $U_i$ uniformly distributed over $\left[0,\Delta V_{pp}\right]$ multiplied by the capacitance coupling $\alpha$~\cite{els_mod}.
\begin{equation}
\nu_{ici} = \alpha \sum_{i=0}^{\seqN} U_i.
\end{equation}
The sum of uniformly distributed independent random variables is Irwin-Hall distributed random variable~\cite{irwin-hall}. Its mean is the sum of the $U_i$ means and its variance is the sum of the $U_i$ variances.
\end{proof}
As was described in~\cite{berman}, when cells are programmed by ISPP, it is possible to compensate ICI errors in the cells that have not reached their target values. If during the write sequence the aggressor cell causes ICI in the victim cell, it can be detected by the verification process of the victim cell leading to canceling excess program steps. However, when a certain cell reached its target level, updating its neighbor cell can still cause ICI.
\begin{proposition}
Let us consider two adjacent $q$-ary cells with additive Gaussian voltage noise $\cN\left(0,\sigma^2\right)$ read by the normalized threshold voltage $V_{ref}$. Let us now assume that in some ISPP operation, the target voltage levels of the two cells are $V_1$ and $V_2$ where $V_2 > V_1>0$. The BER of the victim cell, due to ICI, at the end of the write operation is given by
\begin{equation} \label{eq:BERici}
BER_{ici}\left(\Delta V\right) \simeq \frac{2\left(q-1\right)}{q}Q\left( \frac{V_{ref}-\alpha \Delta V}{\sigma } \right),
\end{equation}
where $\Delta V = V_2 - V_1$, and $\alpha$ is the capacitance coupling between the two cells.
\end{proposition}
\begin{proof}
As was described, ICI effects can be compensated during ISPP operation as long as both cells have not reached their target voltage levels. Therefore, the number of program steps with a potential ICI effect is $\seqN = 2\Delta V / \Delta V_{pp}$. Assuming $\seqN \gg 1$, the ICI noise approaches the Gaussian distribution, and substituting this $\seqN$ in Proposition~\ref{prop:ici_noise} gives the mean and variance of the distribution $\alpha \Delta V$ and $\alpha^2 \Delta V \frac{\Delta V_{pp}+2}{12}$, respectively. As a result, given an additive Gaussian noise $\nu \sim \cN\left(0,\sigma^2\right)$, the total noise $\nu_{total}=\nu+\nu_{ici}$ is the sum of two independent Gaussian random variables distributed as
\begin{equation}
\nu_{total} \sim \cN\left(\alpha \Delta V , \sigma^2+\alpha^2 \Delta V \frac{\Delta V_{pp}+2}{12} \right).\label{eq:v_total}
\end{equation}
To simplify $\sigma_{total}$ in~\eqref{eq:v_total}, we recall that the distribution noise $\nu$ is itself a result of the ISPP pulses raising the voltage level from $0$ to $V_1$, taking $\seqN' =  2V_1 / \Delta V_{pp}$ steps. This gives $\sigma^2 = V_1 \frac{\Delta V_{pp}+2}{12} $ by an argument similar to Proposition~\ref{prop:ici_noise}. In addition, we can take $\frac{\Delta V}{V_1} \leq q$, hence
\begin{equation}
\sigma^2_{total} = \sigma^2+\alpha^2 \frac{\Delta V}{V_1} \sigma^2 \leq \sigma^2 \left(1+q\alpha^2\right).
\end{equation}
Given that $\alpha \ll 1$, we can approximate $\sigma_{total} = \sigma $. Applying the BER calculation of~\eqref{eq:BER} gives~\eqref{eq:BERici}.
\end{proof}
The main conclusion from this ICI model is that ICI errors are more likely when the difference between voltage levels of adjacent cells $\Delta V$ is high. Therefore,~\eqref{eq:BERici} motivates the \wi WOM codes we study here.

\subsubsection{BER improvement}
We now analyze the ICI reduction by using a \wi WOM code. We analyze the worst-case ICI scenario, in which cells with guaranteed \wi are compared with the extreme ICI case: a victim cell in erased state $0$ with neighboring aggressor programmed to level $q-1$.
\begin{theorem}  \label{th:BERfactor}
Using a \wi WOM code on a wordline of $q$-ary memory cells reduces worst-case ICI BER by multiplicative factor
\begin{equation} \label{eq:BERfactor}
\hspace*{-0.1cm}
\exp \left\{ \left(1-\frac{d}{q-1}\right) \frac{\alpha \Delta V}{\sigma} \left( \frac{2b_2 V_{ref}}{\sigma} +b_1 - \left(1+\frac{d}{q-1}\right) \frac{b_2 \alpha \Delta V}{\sigma} \right)   \right\},
\end{equation}
where $\sigma^2$ is the variance of the voltage distribution, $V_{ref}$ is the normalized reference level for read,  $\alpha \Delta V$ is the voltage shift of the victim cell for a worst-case scenario unconstrained write, and $b_1$, $b_2$ are negative constants.
\end{theorem}
\begin{proof}
The maximal voltage shift of the victim cell of an unconstrained write corresponding to updating level $0$ to level $q-1$ is given by $\alpha \Delta V$. By using the \wi WOM codes, the maximal voltage shift is given by $\frac{d}{q-1} \Delta V$. Therefore, by using~\eqref{eq:BERici} we get
\begin{equation} \label{eq:BERcomp}
\frac{BER_{\wiAB}}{BER_{uncon.}}= \frac{ Q\left( \frac{V_{ref}-\alpha \frac{d}{q-1} \Delta V}{\sigma } \right) }{Q\left( \frac{V_{ref}-\alpha \Delta V}{\sigma } \right)}.
\end{equation}
Let us now use the following approximation~\cite{Q} for the $Q\left(x\right)$ function valid for $x \subset \left[0,8\right]$
\begin{equation}
Q\left(x\right) \approx e^{b_2x^2+b_1x+b_0},
\end{equation}
where $b_0,b_1$ and $b_2$ are given by $-0.844, -0.502$ and $-0.469$ respectively. By using this approximation and simplifying,~\eqref{eq:BERcomp} becomes,
\begin{gather}
\frac{BER_{\wiAB}}{BER_{uncon.}}= \exp\left\{b_1\frac{\alpha \Delta V \left(1-\frac{d}{q-1}\right)}{\sigma}\right\} \cdot \nonumber \\
\cdot \exp\left\{ b_2 \frac{2\alpha \Delta V V_{ref}\left(1-\frac{d}{q-1}\right)- \alpha^2 \Delta V^2 \left(1-\frac{d^2}{\left(q-1\right)^2}\right)  }{\sigma^2}  \right\}. \label{eq:BERcomp2}
\end{gather}
Rearranging the terms in~\eqref{eq:BERcomp2} gives~\eqref{eq:BERfactor}.
\end{proof}
By examining Theorem~\ref{th:BERfactor}, we can first notice that the negative term in the exponent $-\left(1+\frac{d}{q-1}\right) \frac{b_2 \alpha \Delta V}{\sigma}$ is negligible due to $\alpha$ which is relatively small. Therefore, we can see that the the ICI BER improvement due to using \wi codes increases exponentially when the imbalance parameter $d$ decreases.
\begin{example}
The BER values for this example are taken from~\cite{ici-ber}, where we assume that these values are also valid for $q=8$. Initial raw BER of $2\cdot 10^{-5}$  yields $\frac{v_{ref}}{\sigma}=4.235$. After $6000$ P/E cycles, the BER becomes $5\cdot 10^{-3}$ due to ICI. That means the normalized voltage shift due to ICI is $\frac{\alpha \Delta V}{\sigma} = 1.472$. Using a the \wi WOM code reduces the worst-case normalized voltage shift to $1.472\frac{d}{q-1}=0.631$. Therefore, the new improved BER is given by
\begin{gather}
BER_{\wiAB} =  \frac{2\left(q-1\right)}{q}Q\left( \frac{V_{ref}-\frac{d}{q-1}\alpha \Delta V}{\sigma } \right)= \\ \nonumber
=\frac{14}{8}Q\left(  4.235 - 0.631\right)= 2.74\cdot 10^{-4}.
\end{gather}
That means, the BER due to ICI was improved by factor $18$ relative to the unconstrained write.
\end{example}
\section{Discussion and Conclusion}
\subsection{Bitline ICI}
We have shown how \wi codes hold a potnetial to significantly reduce ICI within a wordline. This is likely sufficient for the ICI seen in 3D vertical charge-trap flash memories (described in the Introduction). However, standard floating-gate flash memories also suffer from significant {\em bitline} ICI. Therefore, in order to reduce ICI in floating-gate flash memories, the \wi WOM codewords must also be balanced with WOM codewords in adjacent wordlines. We leave this interesting problem as future work.
\subsection{Application to wear leveling}
The same \wi properties suggested here for ICI reduction turn out to be useful for another important problem of flash storage: {\em wear leveling}. Due to limited lifetime of flash cells, it is essential to avoid exceeding the recommended write counts. In order to avoid a scenario in which some pages are worn faster than others, a wear-leveling technique is incorporated to the page mapping layer. This provides good inter-page wear leveling~\cite{EWL}. However, within a page cells can differ significantly in their wear (the total amount of charges written to them so far). These differences may be detrimental to the data reliability, as read/write procedures are commonly tuned to the wear state of the page. Using a WOM code with the \wi property can help equalize the intra-page wear, because no cell will be programmed to a level much higher than the rest of the page.

\subsection{Conclusion}
In this work we presented \wi WOM codes designed to reduce inter-cell interference in multi-level NVMs. Constructions that are simple to implement were given and analyzed. We also derived an upper bound on the number of guaranteed writes of a \wi WOM code and showed that our proposed construction is optimal for some parameters of the code. Lattice-based constructions were also derived and characterized in closed form for $t=2$ writes. Future work can include extending the presented two-cell WOM codes to WOM codes for $n \geq 3$, and the lattice-based codes also to $t\geq 3$.

\section{Acknowledgment}
This work was supported by the Israel Science Foundation, by the Israel Ministry of Science and Technology, and by a GIF Young Investigator grant.

\end{document}